\documentclass[a4paper,12pt]{article}

\usepackage{amsmath}
\usepackage{amssymb}
\usepackage{amsfonts}
\usepackage{graphicx}
 \usepackage{hyperref}
\usepackage{mathrsfs}
 \hypersetup{colorlinks=true,linkcolor=blue,citecolor=blue,linktoc=page}

\renewcommand{\ni}{\noindent}
\newcommand{\eq}{\begin{equation}}
\newcommand{\qe}{\end{equation}}

\newcommand{\N}{\mathbb{N}}                %natural integers
\newcommand{\R}{\mathbb{R}}                     %real numbers
\newcommand{\C}{\mathcal{C}}            %complex numbers

\def\a{\mathbf{a}}
\def\C{\mathscr{C}}
\def\M{\mathscr{R}}
\def\A{\mathbf{A}}
\def\b{\mathbf{b}}
\def\g{\mathbf{g}}
\def\s{\mathbf{s}}
\def\u{\mathbf{u}}
\def\v{\mathbf{v}}

\newcommand{\normTV}[1]{\left\lVert #1\right\lVert_{TV}}
\newcommand{\abs}[1]{\left\vert#1\right\vert}

\newenvironment{proof}{\noindent {\bf Proof:}}{\hfill $\Box$}

\newtheorem{lemma}{Lemma}
\newtheorem{theorem}{Theorem}

\title{\bf Exact solutions to Super Resolution on semi-algebraic domains in higher dimensions}

\begin{document}

\author{Y. de Castro$^1$, F. Gamboa$^2$, D. Henrion$^{3,4,5}$, J.-B. Lasserre$^{3,4}$}

\footnotetext[1]{D\'epartement de Math\'ematiques (CNRS UMR 8628), B\^atiment 425, Facult\'e des Sciences d'Orsay, Universit\'e Paris-Sud 11, F-91405 Orsay Cedex, France.}
\footnotetext[2]{Institut de Math\'ematiques de Toulouse (CNRS UMR 5219). Universit\'e Paul Sabatier, 118 route de Narbonne, 31062 Toulouse, France.}
\footnotetext[3]{CNRS; LAAS; 7 avenue du colonel Roche, F-31400 Toulouse; France.}
\footnotetext[4]{Universit\'e de Toulouse; LAAS; F-31400 Toulouse; France.}
\footnotetext[5]{Faculty of Electrical Engineering, Czech Technical University in Prague,
Technick\'a 2, CZ-16626 Prague, Czech Republic.}

\date{Draft of \today}

\maketitle

\begin{abstract}
We investigate the multi-dimensional Super Resolution problem on closed semi-algebraic domains for various sampling schemes such as Fourier or moments. We present a new {semidefinite programming} (SDP) formulation of the $\ell_1$-minimization in the space of Radon measures in the multi-dimensional frame on semi-algebraic sets.

While standard approaches have focused on SDP relaxations of the dual program (a popular approach is based on  Gram matrix representations), this paper introduces an exact formulation of the primal $\ell_{1}$-minimization exact recovery problem of Super Resolution that unleashes standard techniques (such as {moment-sum-of-squares} hierarchies) to overcome intrinsic limitations of previous works in the literature. Notably, we show that one can exactly solve the Super Resolution problem in dimension greater than $2$ and for a large family of domains described by semi-algebraic sets.
\end{abstract}

\begin{center}\small
\end{center}

\maketitle 

%%%%%%%%%%%%%%%%%%%%%%%%%%%%%%%%%%%%%%%%%%%%%%%%%

\section{Introduction}

\subsection{Super Resolution}
The early formulation of the Super Resolution problem can be identified as the ability of faithfully reconstruct a high-dimensional sparse vector from the observation of a low-pass filter. This situation models important applications in imaging spectroscopy \cite{harris1994super}, image processing \cite{park2003super}, radar imaging \cite{odendaal1994two}, or astronomy \cite{makovoz2005point}. As a theoretical baseline, suppose that one wants to reconstruct a vector $x^\star$ by solving a system of linear equations:
\begin{equation}
\mathbf Ax=\b\quad \mbox{ where }\ x\in\R^{N}\,,\ \b:=\mathbf Ax^\star\,,\ \b\in\R^m\ \mbox{ and }\ m\ll N\,.
\label{lili} 
\end{equation}
If the number $s$ of non-vanishing components of $x^\star$ is small and the matrix $\mathbf A$ enjoys some geometric property, namely the Null Space Property \cite{MR2449058}, that depends only on its kernel and the sparsity $s$, then one can exactly reconstruct $x^\star$ by minimizing  the $\ell_1$-norm within the affine subspace of all solutions of the linear system. Conditions on $m$, $N$, $s$ and properties have been extensively studied, see for example \cite{Dono06a}, \cite{MR2449058} and reference therein. Less than ten year ago, Super Resolution has seeded the ideas of compressed sensing theory \cite{Donoieee06},  \cite{CandTao06}, \cite{CandTao07}.  In this theory,  the matrix $\mathbf A$ is randomized and one is interested both in the construction of probability distributions allowing to show relevant properties, such as the Restricted Isometry Property \cite{CandTao06}, and the stability of the reconstruction process.  This area of research is very fruitful and leads to many practical applications in signal and image processing, see for example \cite{herman09}, \cite{tian07}, or \cite{chen08}. To the best of our knowledge, the first mathematical works on Super Resolution are due to Donoho et al in the early ninety, see \cite{donoho89} and \cite{donoho1992}. In these papers, the term {\it Super Resolution} appeared because the matrix $\mathbf A$ is related to a discretization of some Fourier transform. As a matter of fact, when inverting a discrete Fourier transform, the separation of two close spikes in a sparse signal is made possible by minimizing the $\ell_1$-norm while a linear inversion method is not able to do so. Beside, in these years, many applied researchers were performing Fourier inversion of non negative sparse signals using some {\it entropy} regularization, see for example \cite{Skil87} and \cite{burg67}. At this time, the respective roles of sparsity and non negativity in regard of spikes separation from non linear Fourier inversion methods was not completely clear. An important step for understanding these roles has been taken by lifting the linear equation \eqref{lili} up to the more abstract measure set up, see \cite{Gas90}, \cite{GaGa96} and \cite{DouGa96}:  
\begin{equation}
\langle a_i,\mu\rangle= \int_{X} a_i(x)\mu(dx)=b_i\quad \mbox{ where } \mu\in\M(X)_+\ \mbox{and\ }i=1,\ldots,m\,.
\label{lili1} 
\end{equation}
Here, $(a_i(x))$ is a vector of continuous function defined on $X$, a given compact subset of $\R^n$, $\b=(b_i)\in\mathbb R^m$ and $\M(X)_+$ is the set of all nonnegative {Radon} measures on $X$. In this frame, there exists very special points $\b^\star$ such that the set of all members of  $\M(X)_+$ satisfying (\ref{lili1}) for $\b=\b^\star$ reduces to a singleton $\{\mu_{\b^\star}\}$. Furthermore, $\mu_{\b^\star}$ is a discrete measure concentrated on very few points.  Hence, if one deals with a $\b$ close to a point $\b^\star$ the set of all solutions of \eqref{lili1} (or \eqref{lili} with positivity constraint and a design matrix $A$ discretizing the vectorial function $(a_i(x))_{i=1}^{m}$) is very small. So that, two different methods for selecting a member of this set will lead to similar solutions. We refer to \cite{Ana92}, \cite{GaGa96} and \cite{DouGa96} where quantitative evaluations on the size of the set of solutions is performed and to \cite{GaGa94} and \cite{Lew96} for related evaluations in another context. Notice that the structure of points $\b^\star$ can be completely described in the case where the family of functions $(a_i)$ is a Chebyshev system, $T$-system for short, this includes the case of discrete Fourier transform and moments, see \cite{borwein1995polynomials} or \cite{Stud66} for an exhaustive overview on these systems of functions. A more involved situation is when \eqref{lili1} does not enjoy the non negativity assumption on the measure. This means that one wishes to solve the linear equation:   
\begin{equation}
\langle a_i,\mu\rangle= \int_{X} a_i(x)\mu(dx)=b_i\quad \mbox{ where } \mu\in\M(X)\ \mbox{and\ }i=1,\ldots,m\,.
\label{lili2} 
\end{equation}
Here, $\M(X)$ is the set of {signed Radon} measures on $X$. This is the frame of the present paper.  Surprisingly, as shown by authors of this paper \cite{de2012exact} and under some assumptions on the family of functions $(a_i)$, there exists pairs $(\b^\star, \mu_{\b^\star})\in\R^m\times \M(X)$ such that $b^\star_i=\langle a_i,\mu_{\b^\star}\rangle$, for $i=1,\cdots,m$, and $\mu_{\b^\star}$ is the unique solution of \eqref{lili2} minimizing the total variation norm with $\b=\b^\star$. Such $\mu_{\b^\star}$ are sparse in the sense that they are measures with finite support. 

The study of the solutions to \eqref{lili} and \eqref{lili2} uncovers that $\ell_{1}$-minimization faithfully reconstruct objects concentrated on very few points. However, the analysis in Super Resolution differs dramatically from Compressed Sensing. For instance, {it is well known that sparse $\ell_{1}$-minimization cannot be successful in the ultrahigh-dimensional setting \cite{verzelen2012minimax} where $N\gg \exp(m)$. Indeed it was shown {in} \cite{CGLP} that one needs at least $m\geq (\mathrm{cst}) s\log(N/s)$ measures to faithfully uncover $s$ sparse vectors by $\ell_{1}$-minimization}. Hence, the analysis of Compressed Sensing in terms of high-dimensional random geometry \cite{CGLP} cannot be extended to the space of measure. {Moreover, observe that Compressed Sensing aims at recovering a sparse signal from random projection while, in Super Resolution, the sampling scheme is deterministic}. 

Admittedly their analysis differ but we can bridge the gap between \eqref{lili} and \eqref{lili2} by considering their dual formulations. From the point of view of convex analysis, we see that the dual form of these linear programs aims at reconstructing a dual certificate \cite{CandTao06,de2012exact}, i.e. an $\ell_{\infty}$-constraint linear combination of $(a_{i})$, where $(a_{i})$ are the lines of $\bf A$ in \eqref{lili} and a family of continuous functions in \eqref{lili2}. As pointed out by authors of this paper \cite{de2012exact}, a parallel between Compressed Sensing and Super Resolution exists where the lines of $\bf A$ are the evaluation of a vector of continuous function at some prescribed points. In this frame, Super Resolution can be seen as a Compressed Sensing problem where the dimension $N$ goes to infinity. This analogy persists with the notion of dual certificate, i.e. a solution to the following dual program \eqref{dudu}. Indeed, the dual programs given by the constraints \eqref{lili} and \eqref{lili2} (while minimizing the $\ell_{1}$-norm and the total variation norm respectively) share the same expression:
\eq
\label{dudu}
\sup_{\mathbf u\in\R^{m}}\b^{\top} \u\quad\mathrm{s.t.}\quad\lVert \a^{\top}\u\lVert_{\infty}\leq1\,,
\qe
where $\a=\mathbf A$ in Compressed Sensing \eqref{lili} and  $\a(x)=(a_i(x))_{i=1}^{m}$ is a vector of continuous function in Super Resolution \eqref{lili2}.  Define a \textit{dual certificate} $\mathbf P$ as:
\eq
\label{cece}
\mathbf P=\a^{\top}\u^{\star}\,,
\qe
where $\u^{\star}$ is a solution to the dual program \eqref{dudu}. From the duality properties, note that $\mathbf P$ is a sub-gradient of the $\ell_{1}$-norm at a solution to the primal program \eqref{lili2}. Hence, we can ensure that a target measure $\mu^{\star}$ is a solution to \eqref{lili2} if we are able to construct a dual polynomial \eqref{cece} that interpolates the phases of the weights of $\mu^{\star}$ at its support points. 

From a theoretical point of view, one of the main issue in Super Resolution consists in exhibiting such a {\it dual certificate} $\mathbf P$. In the Fourier frame, notice that an important construction, for target discrete measures whose support satisfy a {\it separable condition}, is given in the fundamental paper \cite{candes2012towards} where a huge step has been taken. Indeed, the authors are the first to give a sharp condition on the support points of the target measure in order to warrant the existence of a dual certificate. Moreover, their proof is based on interpolating, by a Jackson kernel, the phases of the weights of the target measure at its support points and, henceforth, explicitly construct a dual certificate \eqref{cece}. In the Compressed Sensing frame, observe that the same method has been investigated {in} \cite{Kaka11} using a Dirichlet kernel. In the present paper, we will not deal with this issue but rather with the practical resolution of the convex program \eqref{dudu}.

In the Super Resolution frame, remark that the program \eqref{dudu} has finitely many variables but infinitely many constraints.  This last point can be a severe limitation in practice. As a matter of fact, a difficult task is to construct a tractable program that deals with the $\ell_{\infty}$-constraint of \eqref{dudu}. Standard formulations \cite{candes2012towards} are based on Gram matrix representations, see below. Incidentally, these procedures cannot be extended to dimension greater than 2 or to semi-algebraic domains $X$. To cope with this issue, we consider a new parametrization of the primal program based on works of the authors \cite{lasserre}. Note that our method relies on infinitely many parameters but relaxations involving only a finite number of parameters are proved to lead to the exact solution of the primal program. 

\subsection{Previous works} During the last years, theoretical guarantees for exact recovery \cite{bredies2010inverse,de2012exact,candes2012towards}, bounds on the support recovery from inaccurate samplings \cite{azais2013spike,fernandez2013support}, prediction of the Fourrier coefficient from noisy observations \cite{tang2013near}, and noise robustness \cite{duval2013exact} have been showed. These works prove that discrete measures can be recovered, in a robust manner, from few samples using an $\ell_{1}$-method. 

{From a numerical point of view, a solution to $\ell_{1}$-minimization is often computed using the dual program described by \eqref{dudu}. Then, the $\ell_{\infty}$-norm constraint of the dual program \eqref{dudu} is equivalently formulated as a nonnegative constraint on (trigonometric) polynomials. This point of view unleashes Gram matrix representations \cite{candes2012towards} or Toeplitz matrix representations \cite{tang2013near} to handle the constraint of non-negativity of (trigonometric) polynomials on domains. However, these formulations are limited to the frame of the real line and the torus in dimension one (see \cite{dumitrescu2007positive} for instance) since they rely on the Fej\'er-Riesz theorem.
As a matter of fact, the literature of Super Resolution has been focused on the fact that a semidefinite programming (SDP) formulation for the dual problem can be given as long as there exists a spectral factorization for a globally nonnegative trigonometric polynomial \cite{candes2012towards} (Bounded Real Lemma using Fej\'er-Riesz theorem) or a spectral decomposition for semi-definite Toeplitz matrices \cite{tang2013near} (Caratheodory-Toeplitz theorem). Hence, except on the real line and the torus in dimension one, there is no exact SDP formulation of the Super Resolution problem for the dual form. However, relaxed SDP versions of the dual form in dimension greater than $2$ are discussed in \cite{Xual} and they have been used on the $2$-sphere in \cite{bendory2013exact,bendory2014super}.
}

\subsection{Contribution} 
To the best of our knowledge, the present paper is the first to overcome this limitation and expand the scope of Super Resolution implementation to the {\it multi-dimensional} frame {in general basic {\it semi-algebraic domains}}. Indeed, we operate a smart method to tackle numerically the solution of equation (\ref{lili2}) with minimal $\ell_{1}$-norm. This method uses both a re-parametrization in terms of  moment sequences \cite{lasserre} and the so-called {sum-of-squares (SOS}) decompositions of nonnegative multidimensional polynomials, {used widely in systems control theory during the last decade,} see for example \cite{henrion2005positive}. {Notice that, in the scope of Super Resolution, this technique is new and, contrary to other approaches, focuses on the primal program through a truncation of the moment sequences.}

More precisely, given the real numbers $b_i,\;i=1,\ldots,m$, consider the infinite-dimensional optimization problem:
\begin{equation}\label{tv}
\begin{array}{ll}
\inf & \|\mu\|_{TV} \\
\mathrm{s.t.} & \langle a_i, \mu \rangle = b_i, \quad i=1,\ldots,m\\
&\mu\in\M(X),
\end{array}
\end{equation}
where $\|\,.\,\|_{TV}$ is the total variation norm of measures (to be defined later). Notice that, {under
{standard} assumptions, problem} (\ref{tv}) is feasible. That is 
\begin{equation}
\label{a}
 \exists\mu \in \M(X) \mbox{ such that for } i=1,\ldots,m\,,\quad \langle a_i,\mu \rangle = b_i.
\end{equation}
\ni
Our main contribution concerns the numerical resolution of the total variation minimization problem \eqref{tv}. {We extend the univariate $(n=1)$ trigonometric SDP formulation of \cite{candes2012towards} to a
much more general SDP formulation in dimension $n\geq 2$, for measures supported on basic semialgebraic sets.}

To this end, we use the Jordan decomposition of the signed measure $\mu=\mu_+-\mu_-$
as a difference of two nonnegative measures  supported on $X$ {and we follow \cite{lasserre} to define a hierarchy of finite-dimensional primal-dual SDP problems:
\begin{itemize}
\item the primal problems
correspond to SDP relaxations of the conditions that must satisfy finitely many moments of the two nonnegative measures on $X$;
\item the dual  problems correspond to SDP strengthenings using SOS multipliers of the conditions that two
distinguished polynomials are nonnegative on $X$.
\end{itemize}
The moment-SOS hierarchy is indexed by an integer $k$, called relaxation order,
which is the (half of the) number of moments
used to represent the measures in the primal problem, or equivalently, the (half of the) degree of the SOS representations
of the polynomials in the dual problem. The larger is the relaxation order $k$, the larger is the size of the SDP problems,
the number of variables and constraints growing polynomially in $O(k^n)$.

The primal SDP problem features the matrices
of moments of the two nonnegative measures. If the rank of each moment matrix,  as a function of $k$, stabilizes
to a certain constant value, then the corresponding measure is atomic, with the number of atoms equal to the rank.
Therefore, the total variation minimization problem \eqref{tv} has been solved succesfully, and this is certified by the
polynomials solving the SDP problem. Numerical linear algebra can then be used to retrieve the support
of the optimal measure.}

 In the sequel, we present some examples for which our method is the first to give an SDP formulation of the Super Resolution phenomena. As a matter of fact, our procedure encompasses a larger class of measurements than the class of standard moments discussed previously. Our numerical {experiments are carried out} with the {Matlab} interface GloptiPoly 3 which is designed to generate semidefinite relaxations of measure LP problems with polynomial data. So we assume that the functions $a_i(x)$ in LP problem (\ref{p}) are multivariate polynomials, and for notational simplicity, we let $a_i(x):=x^{\alpha_i}=x_1^{\alpha_i^1}\cdots x_n^{\alpha_i^n}$ where $\alpha_i \in {\mathbb N}^n$ are given. Note that the choice of monomials is only motivated for notational simplicity, and that other choices of polynomials (e.g. Chebyshev polynomials) are typically preferable numerically\footnote{A numerical analysis of the impact of the basis is however out of the scope of our work.}. SDP relaxations are then solved with SeDuMi {or MOSEK}, implementations of a primal-dual interior-point algorithm.
{For reproducibility purposes, our Matlab codes (using the public-domain interface GloptiPoly and the SDP solver SeDuMi) of the numerical examples presented next are available for download at}
\begin{center}
{\href{http://homepages.laas.fr/henrion/software/tvsdp.tar.gz}{\tt homepages.laas.fr/henrion/software/tvsdp.tar.gz}}
\end{center}

\subsubsection{Disconnected domain}

\begin{figure}[h!]
\begin{center}
\includegraphics[width=0.8\textwidth]{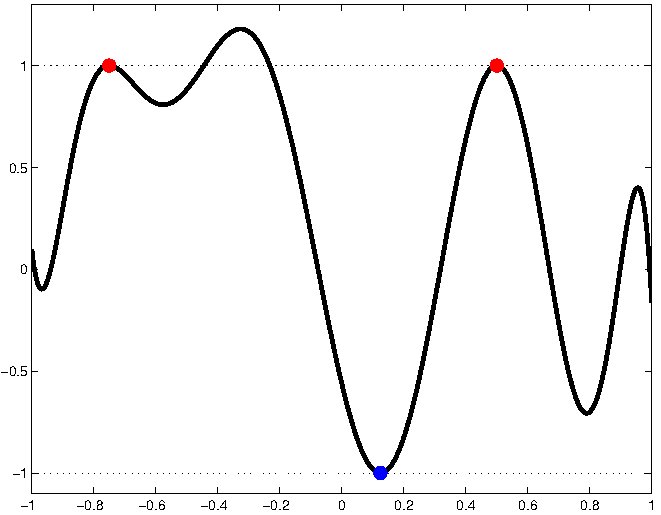}
\end{center}
\caption{Degree 9 polynomial certificate for the univariate example, {with 2 points (red) in the
support of the positive part, and 1 point (blue) in the support of the negative part of the optimal measure.}\label{figex1}}
\end{figure}
We want to recover the measure:
\[
\mu := \delta_{-3/4}+\delta_{1/2}-\delta_{1/8}
\]
on the disconnected set $X:=[-1,-1/2]\cup[0,1]$ which can be modeled as the polynomial superlevel set $X=\{x \in {\mathbb R} \: :\: g_1(x)\geq 0\}$ for the choice:
\[
g_1(x):=-(x+1)(x+1/2)x(x-1).
\]
In LP (\ref{p}) we let $a_i(x):=x^i$ and $b _i:=(-3/4)^i+(1/2)^i-(1/8)^i$ for $i=0,1,2\ldots,9$. 
Solving the SDP relaxation of order {$k=5$} on our standard PC takes 0.2 seconds, and optimality is certified
{from the solution of the primal moment problem} with a rank 2 moment matrix {for $\mu_+$}
and a rank 1 moment matrix {for $\mu_-$}, from which the 3 points can be extracted {using numerical
linear algebra}.
On Figure \ref{figex1} we represent the degree 9 polynomial $\sum_{i=0}^9 u_i x^i$ certifying optimality, {constructed from the solution of the dual SOS problem}. Indeed we can check that the polynomial attains the value $+1$ at the points $x=-3/4$, and $x=1/2$ (in red), it attains the value $-1$ at the point $x=1/4$ (in blue), while taking values between $-1$ and $+1$ on $X$. Notice in particular that the polynomial is larger than $+1$ around $x=-1/4$, but this point is not in $X$. 

\subsubsection{Low-pass filters in dimension greater than $3$}

\begin{figure}[h!]
\begin{center}
\includegraphics[width=0.8\textwidth]{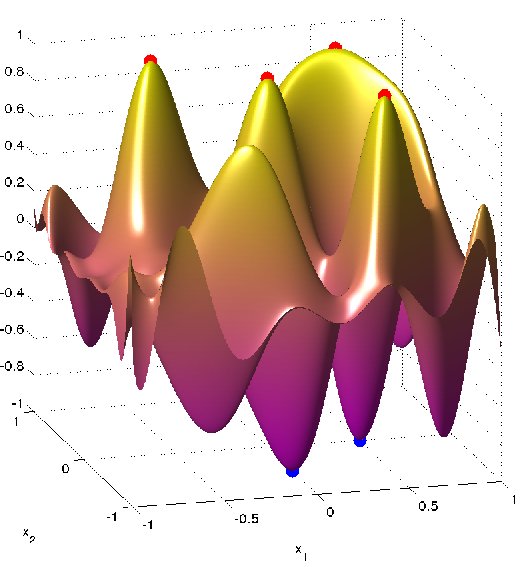}
\end{center}
\caption{Degree 12 polynomial certificate for the bivariate example, {with 4 points (red) in the
support of the positive part, and 2 points (blue) in the support of the negative part of the optimal measure.}
\label{figex2}}
\end{figure}
In the Fourier frame, the recent SDP formulations of $\ell_{1}$-minimization in the space of complex valued measures are based on the Fej\'er-Riesz theorem. As a consequence, they cannot handle dimensions greater than $3$. Observe that our procedure can bypass this limitation. For sake of readability, we present an example in dimension 2 although it can be extended to any dimension. We want to recover the measure
\[
\mu := \delta_{(-1/2,1/2)}+\delta_{(1/2,-1/2)}+\delta_{(1/2,1/2)}+\delta_{(0,0)}
-\delta_{(0,-1/2)}-\delta_{(1/2,0)}
\]
on the box $X:=[-1,1]^2$, {from the knowledge of moments of degree up to 12, i.e. $a_i(x)=x^i$ for $i=0,1,\ldots,12$.}
Solving the SDP relaxation of order {$k=6$} on our standard PC
takes less than 3 seconds, and optimality is certified {from the solution of the primal moment problem}
with a rank 4 moment matrix {for $\mu_+$} for  and a rank 2 moment matrix {for $\mu_-$}
from which the 6 points {of the support of the optimal measure $\mu$}
can be extracted {using numerical linear algebra} with a relative accuracy around $10^{-6}$.
On Figure \ref{figex2} we represent the degree 12 polynomial
certifying optimality, {constructed from the solution of the dual SOS problem}.
Indeed we can check that the polynomial
attains the value $+1$ at the 3 points $x\in\{(-1/2,1/2),(1/2,-1/2),(0,0)\}$, it attains the value $-1$
at the 2 points $x\in\{(0,-1/2),(1/2,0)\}$ (in blue), while taking values between $-1$ and $+1$ on $X$.

\subsubsection{Localization of points on the sphere}

\begin{figure}[h!]
\begin{center}
\includegraphics[width=0.8\textwidth]{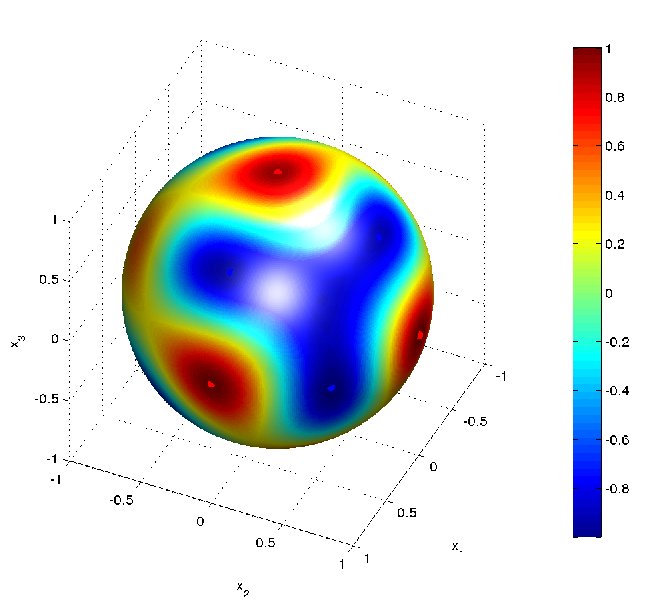}
\end{center}
\caption{Degree {6} polynomial certificate for the sphere example, {with 3 points (red) in the
support of the positive part, and 3 points (blue) in the support of the negative part of the optimal measure.}}
\label{figex3}
\end{figure}

Recent extensions of Super Resolution to spike deconvolution on the $2$-sphere from spherical harmonic measurements has been investigated in \cite{bendory2013exact,bendory2014super}. In these paper, the authors give a sufficient condition for exact recovery using $\ell_{1}$-minimisation and they investigate spikes localization when the measurements are perturbed by additive noise. 

From a numerical point of view, they used a relaxed version of the dual program (bounded real lemma in dimension $d=3$ and a Gram representation of the $\ell_{\infty}$-constraint appearing in the dual). Our work naturally extends to this frame and provides an exact formulation of the primal form.

For sake of numerical code simplicity, we have considered polynomials on $\R^{3}$ restricted to the domain $X$ given by the $2$-sphere (note that one could have used homogenous spherical harmonics instead as in \cite{bendory2013exact}).
{We want to recover the measure}
\[
{
\mu :=\delta_{(1,0,0)}+\delta_{(0,1,0)}+\delta_{(0,0,1)}-
\delta_{(\frac{\sqrt{2}}{2},\frac{\sqrt{2}}{2},0)}-\delta_{(\frac{\sqrt{2}}{2},0,\frac{\sqrt{2}}{2})}-\delta_{(0,\frac{\sqrt{2}}{2},\frac{\sqrt{2}}{2})}
}
\]
{that is supported on the positive orthant just for better visualization purposes.}
{In LP (\ref{p}), the $a_i$ consist of 3-variate monomials of degree up to $5$, i.e. $m=56$.}
Solving the SDP relaxation of order $k=${6} on our standard PC takes less than $20$ seconds, and optimality is certified from the solution of the primal moment problem with a rank {3} moment matrix {for $\mu_+$} and a rank {3} moment matrix {for $\mu_-$}, from which the {6} points can be extracted {using numerical linear algebra}. On Figure \ref{figex3} we represent the degree {6} polynomial certifying optimality on the $2$-sphere, {constructed from the solution of the dual SOS problem}. Indeed we can check that the polynomial
attains the value $+1$ at the 3 prescribed points, it attains the value $-1$ (in red), at the 3 others prescribed points (in blue), while taking values between $-1$ and $+1$ on $X$.

From a theoretical point of view, the minimal separation condition appearing in \cite{bendory2013exact} requires a degree $2$ polynomial. So our example satisfy the sufficient aforementioned condition.

\section{Primal and dual LP formulation}
\subsection{General model and notation}
\label{sugemo}
Let $n$ be a positive integer. Denote by $\R[x]$ the set of all polynomials on $\R^n$, and for $d\in\N$, $\R_d[x]$ the set of all polynomials on $\R^n$ with degree not greater than $d$. Further, we use the following notation:
\begin{itemize}
\item $X\subset \R^n$, is a given closed {basic} semi-algebraic set:
\[
X:= \{ x \in {\mathbb R}^n : g_j(x) \geq 0, \:j=1,\ldots,n_X\}
\]
where $g_j \in {\mathbb R}[x]$, $j=1,\ldots,n_X$, are given polynomials whose degrees are denoted by $d_j$, $j=1,\ldots,n_X$. It is assumed that $X$ is compact with an algebraic certificate of compactness. For example, one of the polynomial inequalities $g_j(x)\geq 0$ should be of the form: 
\[R^2-\sum_{i=1}^n x_i^2 \geq 0\,,
\]
for $R$ a sufficiently large  constant. Let $\g_X :=(g_j)_{j=1,\ldots,n_X}$.
\item Let $\a=(a_{i})_{i=1}^{m}$ be a  linearly independent family of polynomials of degree at most $d$ on $X$. Notice that $m\leq(1+d)^n$. %Notice that  the span of $\a=(a_{i})_{i=1,\ldots,m}$ is $\R_d[x]$. 
\item For monomials we use the multi-index notation
\[
x^{\alpha} := \prod_{j=1}^n x_j^{\alpha_j}
\]
for every $x = (x_1,\ldots,x_n) \in {\mathbb R}^n$ and $\alpha =  (\alpha_1,\ldots,\alpha_n) \in {\mathbb N}^n$. 
\item $\C(X)$, the space of continuous functions on $X$, a Banach space when equipped with the 
sup-norm: 
\[\Vert f\Vert=\sup_{x\in X}\vert f(x)\vert\,.\]

\item $\M(X)$, the space of {signed Radon} measures on $X$,
a Banach space isometrically isomorphic to
the {topological} dual $\C(X)^*$ when equipped with the total variation norm: 
\[\normTV{\mu}=\sup_{\mathcal P}\sum_{E\in\mathcal P}\abs{\mu(E)}\,,\] 
where the supremum is taken over all partitions $\mathcal P$ of $X$ into a finite number of disjoint measurable subsets. 

\item $\C(X)_+\subset \C(X)$ and $\M(X)_+\subset\M(X)$ the respective positive cones of 
nonnegative continuous  functions on $X$ and (nonnegative) {Radon} measures on $X$. We use the 
standard notation
$f\geq0$ and $\mu\geq0$ for membership in $\C(X)_+$ and $\M(X)_+$, respectively.
\item {To denote the integration of a function against a measure, we use} the duality bracket:
\[\langle f,\mu\rangle=\displaystyle\int_X fd\mu\,,\]
for all $f\in \C(X)$, and $\mu\in\M(X)$.
\end{itemize}

\ni
With the usual Jordan decomposition:
\[
\mu=\mu_+-\mu_-
\]
into a sum of two nonnegative Borel measures $\mu_+,\mu_-$,
the optimization problem (\ref{tv})
can be rewritten equivalently as
a linear programming (LP) problem in the convex cone $\M(X)_+$, namely:
\begin{equation}\label{p}
\begin{array}{rcll}
p^* & = & \inf & \langle 1,\mu_+ \rangle + \langle 1,\mu_- \rangle \\
&& \mathrm{s.t.} & \langle a_i, \mu_+ \rangle - \langle a_i, \mu_- \rangle = b_i, \quad i=1,\ldots,m \\
&&& \mu_+ \in \M(X)_+ \\
&&& \mu_- \in \M(X)_+.
\end{array}
\end{equation}
\ni
If  $\a=(a_i)_{i=1,\ldots,m}\in\C(X)^m$ and $\b=(b_i)_{i=1,\ldots,m}\in \R^m$,
problem (\ref{p}) is the dual of the following LP problem in the convex cone ${\mathscr C}(X)_+$:
\begin{equation}\label{d}
\begin{array}{rcll}
d^* & = & \sup & \b^\top \u \\
&& \mathrm{s.t.} & z_+(x) := 1 + \a^\top(x) \u \in {\mathscr C}(X)_+\\
&&& z_-(x) := 1 - \a^\top(x) \u \in {\mathscr C}(X)_+
\end{array}
\end{equation}
where the maximization is w.r.t. $\u=(u_i)_{i=1,\ldots,m} \in {\mathbb R}^m$. Remark that LP problem (\ref{d}) can be also written as:
\[
\begin{array}{rcll}
d^* & = & \sup & \b^\top\u \\
&& \mathrm{s.t.} & \|\a^\top(x)\u\|_{\infty} \leq 1.
\end{array}
\]

\begin{lemma}\label{sd}
There is no duality gap between primal LP (\ref{p}) and dual LP (\ref{d}), i.e. $p^*=d^*$.
\end{lemma}

\begin{proof}
Define the vector $r(\mu_+,\mu_-)\in\R^{m+1}$ by:
\[
r(\mu_+,\mu_-)
:=(\langle 1,\mu_+\rangle + \langle 1,\mu_-\rangle, \:
\langle a_1,\mu_+\rangle - \langle a_1,\mu_-\rangle, \:\ldots\:
\langle a_m,\mu_+\rangle - \langle a_m,\mu_-\rangle)
\]
and the set
\[
R:=\{r(\mu_+,\mu_-) \: :\: (\mu_+,\:\mu_-) \in \M(X)_+\times  \M(X)_+\} \subset {\mathbb R}^{m+1}.
\]
By \cite[Theorem 7.2]{barvinok}, $p^*=d^*$ provided that
$p^*$ is finite and $R$ is closed.
Finiteness of $p^*$ follows from Assumption \ref{a} and nonnegativity
of the objective function $\langle 1,\mu_+\rangle + \langle 1,\mu_-\rangle$.
To prove that $R$ is closed we have to show that for any sequence $(\mu_+^n,\:\mu_-^n)_{n\in\mathbb N} \in \M(X)_+\times \M(X)_+$
such that  $r(\mu_+^n,\mu_-^n) \to \s\in {\mathbb R}^{m+1}$ as $n\to\infty$, one has
$\s=r(\mu,\nu)$ for some finite measures $\mu,\nu\in\M(X)_+$.
Since the supports of all the measures are contained in a compact set,
and since $\langle 1,\mu^n_+\rangle+\langle 1,\mu^n_-\rangle\to s_0$
all measures $\mu^n_+,\mu^n_-$ are uniformly bounded.
Therefore, from the weak-* compactness (and weak-* sequential compactness) of the unit ball (Banach-Alaoglu's Theorem),
there is a subsequence $(\mu_+^{n_k},\:\mu_-^{n_k})_{k\in\mathbb N}$ that converges
weakly-* to an element $(\mu,\:\nu) \in\M(X)_+\times \M(X)_+$.
In particular, as all $a_i$ are continuous,
\[\lim_{k\to\infty} r(\mu_+^{n_k},\:\mu_-^{n_k}) \,=\, r(\mu,\nu),\]
which proves that $R$ is closed.
\end{proof}

\begin{lemma}\label{sa}
For the dual LP problem (\ref{d}) the supremum is attained.
\end{lemma}

\begin{proof}
The feasibility set 
\[
U:=\{\u \in {\mathbb R}^m \: :\: \|\a^\top(x)\u\|_{\infty} \leq 1\}
\]
of the LP problem (\ref{d}) is a closed convex subset of a finite-dimensional
Euclidean space, and it contains the origin. Since the objective function in LP (\ref{d})
is continuous on $U$, the optimum is attained if
$U$ is bounded. Suppose that $U$ is not bounded. Then
there exists a sequence $(\u_n)_{n\in\mathbb N}\subset U$ such that
$\Vert \u_n\Vert\to\infty$ as $n\to\infty$. Write $\u_n=\lambda_n \v_n$, with $\Vert \v_n\Vert=1$.
Notice that $0<\lambda_n\to\infty$ and $\v_n\in U$ because $0\in U$ and $U$ is convex.
Then
\[\|\a^\top(x) \u_n\|_{\infty}\,=\, \|\a^\top(x)\lambda_n \v_n\|_{\infty} \,=\,
= \lambda_n \|\a^\top(x)\v_n\|_{\infty} \leq 1,\]
so that $\|\a^\top(x)\v_n\|_{\infty} \leq \lambda_n^{-1}\to 0$ as $n\to\infty$. Since $\Vert \v_n\Vert=1$,
there exists a subsequence $n_k$ and $\v$ with $\Vert \v\Vert=1$ such that
$\v_{n_k}\to \v$ as $k\to\infty$ and $\|\a^\top(x) \v\|_{\infty} =0$.
By {linear independence of $\a$,} this implies that $\v=0$, a contradiction.
\end{proof}

As a consequence of strong duality of Lemma \ref{sd},
for any optimal primal-dual pair
$(\mu,u)$ we have the complementarity conditions
\[
\langle z_+, \mu_+ \rangle = 0, \quad \langle z_-, \mu_- \rangle = 0
\]
implying jointly with Lemma \ref{sa} that
\[
\mathrm{spt}\:\mu_+ \subset \{\,x \in X : \a^\top(x)\,\u = 1\,\}
\]
and 
\[
\mathrm{spt}\:\mu_- \subset \{\,x \in X : \a^\top(x)\,\u = -1\,\}
\]
for some continuous function $x\mapsto \a^\top(x)\u$, and where $\mathrm{spt}\,\mu$ denotes the {\it support} of $\mu$, that is, the smallest closed set $S\subset \R^n$  such that $\mu(\R^n\setminus S)=0$.

\begin{lemma}
Problem (\ref{tv}) has an optimal atomic measure
supported on at most $2(m+1)$ points.
\end{lemma}

\begin{proof}
Let $\mu_+$ be a nonnegative measure solving problem (\ref{p}) and
let $b_{+i}:=\langle a_i,\mu_+\rangle$, with $a_0:=1$, $i=0,1,\ldots,m$.
If $b_+=0$ then $\mu_+=0$ is trivially atomic (with no atoms), so
assume $b_{+0}\neq 0$,
and consider the probability measure $\bar{\mu}_+:=\mu_+/b_{+0}$ which satisfies
the $m$ equality constraints $\langle a_i,\bar{\mu}_+\rangle = \bar{b}_{+i} :=
b_{+i}/b_{+0}$, $i=1,\ldots,m$. From \cite[Proposition 9.4]{barvinok}
there exists a probability measure $\hat{\mu}_+$ satisfying the
same equality constraints $\langle a_i,\hat{\mu}_+\rangle = \bar{b}_{+i}$
and which is supported on (at most) {$m+1$} points of $X$.
The same reasoning can be applied to any nonnegative measure $\mu_-$ solving problem (\ref{p}),
which has a discrete counterpart $\hat{\mu}_-$ supported on (at most) $m+1$ points of $X$.
The result follows by considering the union of these two discrete supports,
which consists of (at most) $2(m+1)$ points of $X$.
\end{proof}

%%%%%%%%%%%%%%%%%%%%%%%%%%%%%%%%%%%%%%%%%%%%%%%%%%%%%%%%%%%%%%%%%%%%%%

\section{Primal and dual SDP formulation}

Problem (\ref{p}) is an instance of a generalized moment problem.
As such it can be solved by a converging hierarchy of 
finite-dimensional primal-dual semidefinite programming (SDP) problems, 
as described comprehensively in \cite{lasserre}. In the sequel,
we extract the key instrumental ingredients to the construction
of the hierarchy. 

\subsection{Primal moment SDP}

Recall from paragraph \ref{sugemo} that
\[
X := \{x \in \R^n \: : g_j(x) \geq 0, \: j=1,\ldots,n_X\}
\]
is a  basic semi-algebraic set with a straighforward certificate of compactness,
and  let $\g_X :=(g_j)_{j=1,\ldots,n_X}$ denote its defining polynomials.
Given a measure $\mu \in {\mathscr R}_{+}(X)$, the real number
\begin{equation}\label{moment}
y_{\alpha} := \langle x^{\alpha}, \mu \rangle
\end{equation}
is called its {\it moment} of order $\alpha \in \N^n$. Conversely, given
a real valued sequence $y:=(y_\alpha)_{\alpha \in {\mathbb N}^n}$,
if identity \eqref{moment} holds for all $\alpha \in \N^n$,
we say that $y$ has a representing measure $\mu \in {\mathscr R}_{+}(X)$. Equivalently,
sequence $y$ belongs to the infinite-dimensional moment cone
\[
{\mathscr M}(X) := \{(y_{\alpha})_{\alpha \in \N^n} \: :\: y_{\alpha} = \langle x^{\alpha}, \mu \rangle, \:\:
\mu \in {\mathscr R}(X)_+\}.
\]
In the sequel we describe a procedure to approximate this convex cone.

Given $k\in\N$,  let $\R[x]_k$ denote the space of real polynomials of degree at most $k$.
Let us identify a polynomial $p(x) = \sum_{\alpha} p_{\alpha}x^{\alpha} \in \R[x]_k$
with its vector $p$ of coefficients in the monomial basis.
Define the Riesz functional $\ell_y$ as the linear functional acting on polynomials
as follows:
$p \in \R[x]_k \mapsto \ell_y(p) = \sum_{\alpha} p_{\alpha}y_{\alpha} = p^\top y \in \R$.
Note that if sequence $y$ has a representing measure $\mu$, then $\ell_y(p) = \langle p,\mu \rangle$.
Define the {\it moment matrix} of order $k$ as the Gram matrix of the
quadratic form $p \in \R[x]_k \mapsto \ell_y(p^2) \in \R$, i.e. the matrix $M_k(y)$
such that $\ell_y(p^2) = p^\top M_k(y) p$. By construction 
this matrix is symmetric and linear in $y$.
Given a polynomial $g \in \R[x]$, define its localizing matrix of order $k$ as
the Gram matrix of the quadratic form $p \in \R[x]_k \mapsto \ell_y(gp^2) \in \R$, 
i.e. the matrix $M_k(g\,y)$ such that $\ell_y(gp^2) = p^\top M_k(g\,y) p$. 
By construction this matrix is symmetric and linear in $y$.
For $j=1,\ldots,n_X$,
let $k_j$ denote the smallest integer not less than half the degree of polynomial $g_j$,
and let $k_X:=\max\{1,k_1,\ldots,k_{n_X}\}$.
With these notations, and for $k \geq k_X$, define the finite-dimensional moment cone
\[
{\mathscr M}_k(\g_X) := \{(y_{\alpha})_{|\alpha|\leq 2k} \: :\: M_k(y) \succeq 0, \: M_{k-k_j}(g_j\,y) \succeq 0, \:
j=1,\ldots,n_X\}
\]
where $\succeq 0$ means positive semidefinite. 

Let $y_+$ resp. $y_-$ denote the sequence of moments
\[
y_{+\alpha} := \int x^{\alpha} \mu_+(dx), \quad y_{-\alpha} := \int x^{\alpha} \mu_-(dx)
\]
of $\mu_+$ (resp. $\mu_-$), indexed by $\alpha \in {\mathbb N}^n$.
Primal measure LP (\ref{p}) can be written as a primal moment LP:
\[
\begin{array}{rcll}
p^* & = & \min & y_{+0} + y_{-0} \\
& & \mathrm{s.t.} & \A(y_+,y_-) = \b \\
& & & y_+ \in {\mathscr M}(X)\\
& & & y_- \in {\mathscr M}(X) \\
\end{array}
\]
where the linear system of equations $\A(y_,y_-) = \b$ models the linear moment constraints.
The moment relaxation of order $k\geq \max\{k_X,d\}$ of the primal  moment LP then reads:
\begin{equation}\label{psdp}
\begin{array}{rcll}
p^*_k & = & \min & y_{+0} + y_{-0} \\
& & \mathrm{s.t.} & \A(y_+,y_-) = \b \\
& & & y_+ \in {\mathscr M}_k(\g_X)\\
& & & y_- \in {\mathscr M}_k(\g_X) \\
\end{array}
\end{equation}
where the minimization is w.r.t. a vector $(y_+,y_-)$ of moments of degree at most $2k$.
For fixed $k$, problem (\ref{psdp}) is a finite-dimensional linear programming problem in the convex cone
of positive semidefinite matrices, i.e. an SDP problem. When $k$ varies, the number of moments, as well as the size
of the moment and localizing matrices in problem (\ref{psdp}) are binomial coefficients growing in $O(k^n)$.

It can be shown that $(p^*_k)$ is a monotonically nondecreasing converging sequence
of lower bounds on $p^*$, i.e. $p^*_{k+1} \geq p^*_k$ and $\lim_{k\to\infty} p^*_k=p^*$.
However, in the context of solving LP (\ref{p}), a more relevant result is the following:

\begin{theorem}\label{rank}
For a given relaxation order $k \geq \max\{d,k_X\}$, let $(y^*_+,y^*_-)$ denote the solution of the moment SDP (\ref{psdp}).
If
\begin{equation}
\label{rank-cond}
\mathrm{rank}\:M_{k-k_X}(y^*_+) = \mathrm{rank}\:M_{k}(y^*_+) \:\:\text{and}\:\:
\mathrm{rank}\:M_{k-k_X}(y^*_-) = \mathrm{rank}\:M_{k}(y^*_-)
\end{equation}
then $p^*_k=p^*$ and LP (\ref{p}) has an optimal solution $(\mu^*_+,\mu^*_-)$
with $\mu^*_+$ (resp. $\mu^*_-$) atomic supported at $r_+:=\mathrm{rank}\:M_{k}(y^*_+),$ (resp. 
$r_-:=\mathrm{rank}\:M_{k}(y^*_-)$) points.
\end{theorem}
\begin{proof}
{By  \cite[Theorem 3.11]{lasserre}, $y^*_+$ (resp. $y^*_-$) is the vector of moments up to order 
$2k$, of a measure $\mu^*_+$ (resp. $\mu^*_-$) supported on
$\mathrm{rank}\:M_{k}(y^*_+)$ points (resp. $\mathrm{rank}\:M_{k}(y^*_-)$) points of $X$. Therefore
$(\mu^*_+,\mu^*_-)$ is a feasible solution of (\ref{p}) with value $p_k\leq p^*$, which proves that $(\mu^*_+,\mu^*_-)$ is 
an optimal solution of (\ref{p}) and $p_k=p^*$.
}
\end{proof}
\vspace{0.2cm}

Given a moment matrix $M_{k}(y^*_+)$ satisfying the rank constraint of Theorem \ref{rank},
there is a numerical linear algebra algorithm that extracts the $r_+$ points of the support of
the corresponding atomic measure $\mu^*_+$, and similarly for $\mu^*_-$.
The algorithm is described e.g. in \cite[Section 4.3]{lasserre} and it is implemented in the
Matlab toolbox GloptiPoly 3.

A certificate of optimality can be obtained by solving the dual problem to primal SDP
problem (\ref{psdp}), and this is described next.

\subsection{Dual SOS SDP}

For a given integer $k$, let $\Sigma[x]_k \subset \R[x]_{2k}$ denote the space
of SOS (sums of squares) polynomials of degree at most $2k$. If $p \in \Sigma[x]_k$
this means that there exists $q_j \in \R[x]_k$, $j=1,\ldots,n_p$, such
that $p=\sum_{j=1}^{n_p} q^2_j$.
Let
\[
{\mathscr P}(X) := \{p \in \R[x] \: :\: p(x) \geq 0, \: \forall x \in X\}
\]
denote the infinite-dimensional cone of nonnegative polynomials on $X$,
and for $k \geq k_X$,  define the finite-dimensional SOS cone, also called quadratic module
\[
{\mathscr P}_k(\g_X) := \{p_0 + \sum_{j=1}^{n_X} g_j p_j, \: p_j \in \Sigma[x]_{k-k_j}, \: j=0,1,\ldots,n_X\} \subset \R[x]_{2k}.
\] 
Under the above assumptions on the polynomial family $\g_X$ defining $X$,
from Putinar's theorem, {see e.g. \cite[Theorems 2.14 and 3.8]{lasserre},}
it holds {that ${\mathscr P}(X) \cap {\mathbb R}[x]_\kappa$ is the closure of
$\left(\cup_{k\geq k_X} {\mathscr P}_k(\g_X)\right) \cap {\mathbb R}[x]_\kappa$ for all $\kappa \geq k_X$.}
Observe also that ${\mathscr M}(X)$ (resp. ${\mathscr M}_k(\g_X)$) is the dual cone to ${\mathscr P}(X)$
(resp. ${\mathscr P}_k(\g_X)$). Whereas testing whether a given polynomial belongs to ${\mathscr P}(X)$
is a difficult task, testing whether a given polynomial belongs to ${\mathscr P}_k(\g_X)$, for a fixed $k$, amounts
to solving an SDP problem.

Dual continuous function LP (\ref{d}) can be written as a positive polynomial LP
\[
\begin{array}{rcll}
d^* & = & \max & \b^\top \u \\
& & \mathrm{s.t.} & 1 + \a^\top(x)\,\u \in {\mathscr P}(X) \\
& & & 1 - \a^\top(x)\,\u \in {\mathscr P}(X)
\end{array}
\]
and its SOS strengthening of order $k\geq k_X$ reads:
\begin{equation}\label{dsdp}
\begin{array}{rcll}
d^*_k & = & \max & \b^\top \u \\
& & \mathrm{s.t.} & 1 + \a^\top(x)\,\u \in {\mathscr P}_k(\g_X)\\
& & & 1 - \a^\top(x)\,\u \in {\mathscr P}_k(\g_X)\\
\end{array}
\end{equation}
where the maximization is w.r.t. a vector $\u \in {\mathbb R}^m$.
It turns out that this is SOS problem (\ref{dsdp}) is an SDP problem dual to
the moment problem (\ref{psdp}):
\begin{lemma}\label{duality}
There is no duality gap between SDP problems (\ref{psdp}) and (\ref{dsdp}),
i.e. $p^*_k=d^*_k$, and both (\ref{psdp}) and (\ref{dsdp}) have an optimal solution.
\end{lemma}
\begin{proof}
{We first show that (\ref{psdp}) has an optimal solution.
Recall that one of constraints $g_j(x)\geq0$ that define $X$ states that $M-\Vert x\Vert^2\geq0$ for some $M>1$.
From the constraint $M_{k-k_j}(g_j y_+)\succeq0$ one deduces that $\ell_{y_+}(M-x_i^2)\geq0$, and
$\ell_{y_+}(Mx_i^t -x_i^{t+2})\geq0$ for every $t=1,\ldots,2k-2$. Hence 
$\ell_{y^+}(x_i^{2k})\leq M^{k}y_{+0}$ for every $i=1,\ldots,n$.
With similar arguments, $\ell_{y^-}(x_i^{2k})\leq M^{k}y_{-0}$ for every $i=1,\ldots,n$.
By \cite[Proposition 3.6]{lasserre} $\vert y_{+\alpha}\vert \leq M^ky_{+0}$ and
$\vert y_{-\alpha}\vert \leq M^ky_{-0}$ for all $\alpha\in\mathbb{N}^n_{2k}$. 
Next, in a minimizing sequence $(y^s_+,y^s_-)$, $s\in\mathbb{N}$, of (\ref{psdp}) one has 
$y^s_{+0}+y^s_{-0}\leq y^1_{+0}+y^1_{-0}=:\rho$ for all $s$, and so
$\vert y^s_{+\alpha}\vert \leq M^k\rho$ and $\vert y^s_{-\alpha}\vert \leq M^k\rho$ for all $\alpha\in\mathbb{N}^n_{2k}$, and all $s=1,\ldots$. 
From this we deduce that there is a subsequence $(y^{s_t}_+,y^{s_t}_-)$, $t\in\mathbb{N}$, that converges to
some $(y^*_+,y^*_-)$ as $t\to\infty$, with value $y^*_{+0}+y^*_{-0}=p^*_k$. In addition by a simple continuity argument,
$M_k(y^*_+)\succeq0$ and $M_{k-k_j}(g_j \,y^*_+)\succeq0$, $j=1,\ldots,n_X$. Similarly
$M_k(y^*_-)\succeq0$ and $M_{k-k_j}(g_j\, y^*_-)\succeq0$, $j=1,\ldots,n_X$, which proves that
$(y^*_+,y^*_-)$ is an optimal solution of (\ref{psdp}). 

Next, the set of optimal solutions $y^*:=\{(y^*_+,y^*_-)\}$ of (\ref{psdp}) is compact. This follows from 
$\vert y^*_{+\alpha}\vert \leq M^ky^*_{+0}\leq M^kp^*_k$ and
$\vert y^*_{-\alpha}\vert \leq M^ky^*_{-0}\leq M^kp^*_k$ for all $\alpha\in\mathbb{N}^n_{2k}$. And so
every sequence in $y^*$ has a converging subsequence. From \cite[Chapter IV. Theorem 7.2]{barvinok}
one also deduces that there is no duality gap
between (\ref{psdp}) and (\ref{dsdp}). 

It remains to prove that (\ref{dsdp}) has an optimal solution. Consider a maximizing sequence 
$(\mathbf{u}_t)_{t\in\mathbb{N}}$, with $\mathbf{b}^T\mathbf{u}_t\to p^*_k=p^*$ as $t\to\infty$.
By feasibility in (\ref{dsdp}), one has $\Vert \mathbf{a}(x)^T\mathbf{u}_t\Vert_\infty\leq 1$ for all $t$ and therefore
$(\mathbf{u}_t)\subset \mathcal{U}:=\{\mathbf{u}\in\mathbb{R}^n:\Vert \mathbf{a}(x)^T\mathbf{u}\Vert_\infty\leq 1\}$ 
and $\mathcal{U}$ is compact (see the proof of Lemma \ref{sa}).
Therefore there exists $\mathbf{u}^*\in\mathcal{U}$ and a subsequence $(t_\ell)_{\ell \in \mathbb N}$
such that  $\mathbf{u}_{t_\ell}\to \mathbf{u}^*\in\mathcal{U}$ as $\ell\to\infty$.
In particular $\mathbf{b}^T\mathbf{u}^*=p^*$. Moreover, since by Lemma \ref{closed} in the Appendix the convex cone $\mathscr{P}_k(\mathbf{g}_X)$ is closed,
$1-\mathbf{a}(x)^T\mathbf{u}_{t_\ell}\to1-\mathbf{a}(x)^T\mathbf{u}^*\in\mathscr{P}_k(\mathbf{g}_X)$, which proves that
$\mathbf{u}^*$ is an optimal solution of (\ref{dsdp}).
}
\end{proof}

Assume that the rank conditions of Theorem \ref{rank} is satisfied
at some relaxation order $k$, and let $(\mu^*_+,\mu^*_-)$ denote the atomic
measures optimal for problem (\ref{p}), obtained from the solution
of the primal SDP problem (\ref{psdp}). Let $\mathbf{u}^*$ denote an optimal solution
of the dual SDP problem (\ref{dsdp}). The duality result of Lemma \ref{duality}
implies that
\[
{\mbox{Supp}}\:\mu^*_+ \subset \{x \in X : \a^\top(x)\mathbf{u}^* = 1\}
\]
and
\[
{\mbox{Supp}}\:\mu^*_- \subset \{x \in X : \a^\top(x)\mathbf{u}^* = -1\}
\]
so that the polynomial $\a^\top(x)\mathbf{u}^*$ can be used as a certificate
of optimality.
We formulate this in the following dual to Theorem \ref{rank}.

\begin{lemma}
Assume that the rank conditions (\ref{rank-cond}) of Theorem \ref{rank} hold.
Let us denote by $\u^*$ the optimal solution of SOS SDP (\ref{dsdp}). Then
the polynomial $z^*_+(x):=1+\a^\top(x)\u^*$ vanishes at the $r_+$ points
of the support of $\mu^*_+$, and the polynomial $z^*_-(x):=1-\a^\top(x)\u^*$ vanishes
at the $r_-$ points of the support of $\mu^*_-$.
\end{lemma}
\begin{proof}
Let us denote by $\{x^k_+\}_{k=1,\ldots,r_+} \subset X$ the points of the support of
the optimal measure $\mu^*_+$, computed from the moments $y^*_+$ solving
optimally moment SDP (\ref{psdp}).
By complementarity of the solutions of primal-dual SDP (\ref{psdp}) and (\ref{dsdp}),
it holds $\langle z^*_+,\mu^*_- \rangle = 0$ and hence
$\langle z^*_+,\delta_{x^k_+} \rangle = z^*_+(x^k_+) = 0$ for each $k=1,\ldots,r_+$.
The proof is similar for $z^*_-$ and $\mu^*_-$.
\end{proof}

\section{Discussion}

We would like to point out that the developments in this paper were
inspired by a previous work on optimal control for linear systems
formulated as a primal LP \eqref{p} on measures and a dual LP
on continuous functions (\ref{d}), and solved numerically with primal-dual
moment-SOS SDP hierarchies \cite{claeys13,claeys14}. Formulating optimal control
problems as moment problems was a classical research topic in the 1960s,
where optimal control laws were sought in measures spaces (completions
of Lebesgue spaces) to allow for oscillations and concentrations, see e.g.
\cite{krasovski68} or the overview in \cite[Section III]{fattorini99}. In the case
of linear optimal control of an ordinary differential equation of order $n$,
it was proved in \cite{neustadt64} that there is always an $n$-atomic
optimal measure solving problem (\ref{p}).

In practice, Theorem \ref{rank} should be used as follows:
\begin{itemize}
\item[1] Let $k = \max\{d,k_X\}$.
\item[2] Solve SDP problem (\ref{psdp}) and its dual (\ref{dsdp}) with a primal-dual algorithm.
\item[3] If the rank condition (\ref{rank-cond}) of Theorem \ref{rank} is satisfied, then extract
the measure from the solution of (\ref{psdp}) and the polynomial certificate
from the solution of (\ref{dsdp}). Otherwise, let $k=k+1$, and go to 1.
\end{itemize}

We conjecture that if the data $\a, \b$ in problem (\ref{p})
are generic, then there is a finite value of $k$ for which the
rank condition of Theorem \ref{rank} is satisfied. The rationale behind this assertion follows from a result by Nie \cite{nie}
on {\it generic} finite convergence for the moment-SOS SDP hierarchy for polynomial optimization over 
compact basic semi-algebraic sets. Translated in the present context for a fixed family of  data
$\mathbf{a}$, results in \cite{nie} yield that
there is a set of polynomials $\{h_1,\ldots,h_L\}\subset \mathbb{R}[\mathbf{u}]$, such that,
given a feasible solution $\mathbf{u}$ of (\ref{d}),
if $h_\ell(\mathbf{u})\neq0$ for all $\ell=1,\ldots,L$, then indeed
\[1+\mathbf{a}^T(x)\mathbf{u}\,=\,p^1_0(x)+\sum_{j=1}^{n_X}p^1_j(x)\,g_j(x),\qquad x\in\mathbb{R}^n,\]
and 
\[1-\mathbf{a}^T(x)\mathbf{u}\,=\,p^2_0(x)+\sum_{j=1}^{n_X}p^2_j(x)\,g_j(x),\qquad x\in\mathbb{R}^n,\]
for some SOS polynomials $p^k_j$, $k=1,2$ and $j=1,\ldots,n_{X}$. So if the optimal solution $\mathbf{u}^*$ of (\ref{d})
satisfies $h_\ell(\mathbf{u}^*)\neq0$, $\ell=1,\ldots,L$, then $d^*=d^*_k$ for some index $k$ (i.e. finite convergence takes place).
Similarly, by \cite{nie2} the rank-condition (\ref{rank-cond}) of Theorem \ref{rank} also holds generically for 
polynomial optimization (which however is a context different from the present context). Put differently,
finite convergence would not hold only if every optimal solution $\mathbf{u}$ of (\ref{d}) would be a zero of some polynomial 
of the family $\{h_1,\ldots,h_L\}\subset\mathbb{R}[\mathbf{u}]$.
But so far we have not proved that at least one optimal solution $\mathbf{u}^*$  of (\ref{d}) is not a zero of some
of the polynomials $h_\ell$, at least for generic $\mathbf{b}$.

Of course,  finite convergence occurs for trigonometric polynomials 
on $X=[0,2\pi]$, which follows from the Fej\'er-Riesz theorem and this was exploited in the 
landmark paper \cite{candes2012towards}. Similarly, but apparently not so well-known, 
the Fej\'er-Riesz theorem also holds in dimension $n=2$.  
Indeed it follows from Corollary 3.4 in \cite{scheiderer2006sums} that every non-negative bivariate trigonometric polynomial can be written as a sum of squares of trigonometric polynomials\footnote{We are grateful to Markus Schweighofer for providing this reference.}. 
So again for trigonometric polynomials on $X=[0,2\pi]^2$, finite convergence of the hierarchy (\ref{dsdp}) takes place, i.e., $d^*_k=d^*$.
Note however that in contrast to the one-dimensional case,
there is {\it no} explicit upper bound on the degrees of the sum of squares which are required, so that
even in the two-dimensional Fourier case we do not have an a priori estimates on the smallest value of $k$
for which $d^*_k=d^*$ an for which we can guarantee that the rank condition of Theorem \ref{rank} is satisfied.

{Generally speaking}, even if {our genericity} conjecture is true,
we do not have a priori estimates on the smallest value of $k$
for which Theorem \ref{rank} holds. As mentioned above,
this also true even in the two-dimensional case on $[0,2\pi]^2$ where finite convergence is guaranteed
in all  cases.

\section{Appendix}
We first recall some standard results of convex analysis.
\begin{lemma}(\cite[Corollary I.1.3]{faraut})
Let $C\subset\mathbb{R}^n$ be a closed convex cone with dual $C^*=\{y: \langle x,y\rangle \geq0,\:\forall x\in C\}$.
Then ${\rm int}\:C^*\neq\emptyset\Leftrightarrow C\cap (-C)=\{0\}$.
\end{lemma}
\begin{lemma}(\cite[Corollary I.1.6]{faraut})
\label{aux1}
Let $C\subset\mathbb{R}^n$ be a closed convex cone whose dual $C^*$ has nonempty interior. Then 
for all $y\in{\rm int}\:C^*$, the set $\{x\in C: \langle x,y\rangle \leq 1\}$ is compact.
\end{lemma}

\begin{lemma}
\label{closed}
The convex cone $\mathscr{P}_k(\mathbf{g}_X)$ is closed.
\end{lemma}
\begin{proof}
Let $\mathcal{S}^n_+$ be the convex cone of real symmetric matrices of size $n$ that are positive semidefinite.
{Let $\mathbb{N}^n_k$ be the set
of $n$-dimensional integer vectors $\alpha$ such that $\sum_{i=1}^n \alpha_i  \leq k$} and let
$v_k(x):=(x^\alpha)_{\alpha\in\mathbb{N}^n_k}$ be a vector of monomials {of degree up to $k$.}
Next let $v_k(x)\,v_k(x)^T=\sum_{\alpha\in\mathbb{N}^n_{2k}}x^\alpha\,A_{0\alpha}$ and 
\[v_{k-v_j}(x)\,v_{k-v_j}(x)^T\,g_j(x)\,=\,\sum_{\alpha\in\mathbb{N}^n_{2k}}x^\alpha\,A_{j\alpha},\quad j=1,\ldots,n_X,\]
for some appropriate real symmetric matrices $A_{j\alpha}$.

Consider a sequence
$(q_t)_{t\in\mathbb N} \subset\mathscr{P}_k(\mathbf{g}_X)$ such that $q_t\to q\in \mathbb{R}[x]_{2k}$ as $t\to\infty$. That is
\[q_t(x)\,=\,p_{0t}(x)+\sum_{j=1}^{n_X}p_{jt}(x)\,g_j(x),\qquad\forall x\in\mathbb{R}^n,\]
for some $p_{jt}\in\Sigma[x]_{k-v_j}$, $j=0,\ldots,n_X$, for all $t \in \mathbb N$. More precisely, coefficient-wise
\begin{equation}
\label{aux44}
q_{t\alpha}\,=\,\langle Q_{0t},A_{0\alpha}\rangle +
\sum_{j=1}^{n_X}\langle Q_{jt},A_{j\alpha}\rangle,\qquad\,\forall \, \alpha\in\mathbb{N}^n_{2k},\end{equation}
for some appropriate matrices $Q_{jt}\in\mathcal{S}^{k-v_j}_+$.
Let $y=(y_\alpha)_{\alpha\in\mathbb{N}^n_{2k}}$ be the moments $y_\alpha:=\int_X x^\alpha dx$ of the measure uniformly supported on $X$.
Observe that since $X$ has nonempty interior,
\[\int_X p(x)\,dx\,>\,0\qquad\forall 0\,\neq p\in\Sigma[x]_k,\]
and
\[\int_X p(x)\,g_j(x)\,dx\,>\,0\qquad\forall 0\,\neq p\in\Sigma[x]_{k-v_j}, \quad j=1,\ldots,n_X.\]
Put differently $M_k(y)\succ0$ and $M_{k-v_j}(g_j\,y)\succ0$, $j=1,\ldots,n_X$.

The convergence $q_t\to q$ implies $\langle q_t,y\rangle \to \langle q,y\rangle$ as $t\to\infty$. 
Hence there is some $\eta$ such that $\eta\geq\,\langle q_t,y\rangle$ for all $t \in \mathbb N$. This in turn implies
\begin{eqnarray}
\nonumber
\eta\,\geq\,\langle q_t,y\rangle&=&\langle p_{0t},y\rangle +\sum_{j=1}^{n_X}\langle p_{jt}\,g_j,y\rangle\\
\label{aux33}
&=&\langle Q_{0t}, M_k(y)\rangle +\sum_{j=1}^{n_X}
\langle Q_{jt}, M_{k-v_j}(g_j \,y)\rangle.
\end{eqnarray}
Therefore
\[\sup_t\,\langle Q_{0t}, M_k(y)\rangle \leq\,\eta,\quad 
\sup_t \langle Q_{jt}, M_{k-v_j}(g_j \,y)\rangle\,\leq\, \eta,\quad j=1,\ldots,n_X.\]
As $0\prec M_{k}(y)\,\in\,{\rm int}(\mathcal{S}^{k}_+)^*$, and
$0\prec M_{k-v_j}(g_j \,y)\,\in\,{\rm int}(\mathcal{S}^{k-v_j}_+)^*$, $j=1,\ldots,n_X$,
one may invoke Lemma \ref{aux1} and conclude that the sequences
$(Q_{0t})_{t\in\mathbb N}\subset\mathcal{S}^k_+$ and
$(Q_{jt})_{t\in\mathbb N}\subset\mathcal{S}^{k-v_j}_+$ are norm-bounded. Therefore there is a subsequence 
$(t_\ell)_{\ell\in\mathbb{N}}$ and matrices $Q_0\in\mathcal{S}^k_+$
and $Q_j\in\mathcal{S}^{k-v_j}_+$, $j=1,\ldots,n_X$, such that
\[Q_{0t_\ell}\to Q_0,\quad Q_{jt_\ell}\to Q_j,\quad j\,=1,\ldots,n_X\]
as $\ell\to\infty$.
Taking the limit for the subsequences $(q_{t_\ell\alpha})_{\ell\in\mathbb N}$ and 
$(Q_{jt_\ell})_{\ell\in\mathbb N}$ in (\ref{aux44}) yields coefficient-wise
\[q_{\alpha}\,=\,\langle Q_0,A_{0\alpha}\rangle +
\sum_{j=1}^{n_X}\langle Q_j,A_{j\alpha}\rangle,\qquad\,\forall \, \alpha\in\mathbb{N}^n_{2k},\]
which proves that $q\in\mathscr{P}_k(\mathbf{g}_X)$, the desired result.
\end{proof}

\end{document}